\documentclass[journal,draftcls,onecolumn,12pt,twoside]{IEEEtran}

\usepackage[T1]{fontenc} 
\usepackage{times}
\usepackage{amsmath}
\usepackage{amssymb}
\usepackage{amsthm}
\usepackage{color}
\usepackage{algorithm}
\usepackage[noend]{algorithmic}
\usepackage{graphicx}
\usepackage{subfigure}
\usepackage{multirow}
\usepackage[bookmarks=false,colorlinks=false,pdfborder={0 0 0}]{hyperref}
\usepackage{cite}
\usepackage{bm}
\usepackage{arydshln}
\usepackage{mathtools}
\usepackage{microtype}
\usepackage{subfigure}
\usepackage{float}
\usepackage[figuresright]{rotating}
\usepackage{threeparttable}
\usepackage{booktabs}

\newtheorem{theorem}{Theorem}

\newtheorem{corollary}[theorem]{Corollary}

\long\def\symbolfootnote[#1]#2{\begingroup
	\def\thefootnote{\fnsymbol{footnote}}\footnote[#1]{#2}\endgroup}
\renewcommand{\paragraph}[1]{{\bf #1}}

\long\def\symbolfootnote[#1]#2{\begingroup
	\def\thefootnote{\fnsymbol{footnote}}\footnote[#1]{#2}\endgroup}

\ifodd1
\newcommand{\com}[1]{\textbf{\color{blue} (COMMENT: #1)}}
\else\newcommand{\com}[1]{}\fi

\begin{document}
	\title{Generalized Simple Regenerating Codes: Trading Sub-packetization and Fault Tolerance}
	
	\author{Zhengyi Jiang, Hao Shi, Zhongyi Huang, Bo Bai, Gong Zhang and Hanxu Hou
 \thanks{Z. Jiang, H. Shi and Z. Huang are with 
the Department of Mathematics Sciences, Tsinghua University, Beijing, China~(E-mail: jzy21@mails.tsinghua.edu.cn, shih22@mails.tsinghua.edu.cn, zhongyih@tsinghua.edu.cn). B. Bai and G. Zhang are with the Theory Lab, Central Research Institute, 2012 Labs, Huawei Tech. Co. Ltd., Hong Kong SAR~(E-mail: baibo8@huawei.com, nicholas.zhang@huawei.com).
H. Hou is with the School of Electrical Engineering \& Intelligentization, Dongguan University of Technology~(E-mail: houhanxu@163.com).
This work was partially supported by the National Key R\&D Program of
China (No. 2020YFA0712300), the National Natural Science Foundation of China (No. 62071121, 61871136, 12025104),
Basic Research Enhancement Program of China under Grant 2021-JCJQ-JJ-0483.}}
	
	\maketitle
	
	\begin{abstract}
		 Maximum distance separable (MDS) codes have the optimal trade-off between storage efficiency and fault tolerance, which are widely used in distributed storage systems. As typical non-MDS codes, simple regenerating codes (SRCs) can achieve both smaller repair bandwidth and smaller repair locality than traditional MDS codes in repairing single-node erasure.
		 
In this paper, we propose {\em generalized simple regenerating codes} (GSRCs) that can support much more parameters than that of SRCs. We show that there is a trade-off between sub-packetization and fault tolerance in our GSRCs, and SRCs achieve a special point of the trade-off of GSRCs. We show that the fault tolerance of our GSRCs increases when the sub-packetization  increases linearly. We also show that our GSRCs can locally repair any singe-symbol erasure and any single-node erasure, and the repair bandwidth of our GSRCs is smaller than that of the existing related codes.
		
	\end{abstract}
	
	
	\IEEEpeerreviewmaketitle
	
	\section{Introduction}
	\label{sec:intro}

An $(n,k,m)$ array code encodes a data file of $km$ {\em data symbols} to
obtain $(n-k)m$ {\em coded symbols} such that the total $nm$ symbols are stored in $n$ nodes with each node
storing $m$ symbols, where $k < n$ and $m\geq 1$. The number of symbols stored in
each node, i.e., the size of $m$, is called {\em sub-packetization level}.
The codes are {\em maximum distance separable (MDS)} codes if any $k$ out of $n$ nodes can retrieve all
$km$ data symbols.
{\em Repair bandwidth} defined as the number of symbols downloaded from the helper nodes in repairing one single erased node is an important metric in designing the codes. It is shown in \cite{dimakis2010} that the minimum repair bandwidth of $(n,k,m)$ MDS array codes by connecting $d$ helper nodes ($d\geq k$) is $\frac{dm}{d-k+1}$ symbols, and the codes achieving the above minimum repair bandwidth for each node are called {\em minimum storage regenerating} (MSR) codes. However, explicit constructions of high code-rate (i.e., $\frac{k}{n}>\frac{1}{2}$) MSR codes require extensively large sub-packetization \cite{2018A}. It is of practical significance to design array codes that have lower sub-packetization, lower repair bandwidth and are easy to implement.
		
Simple regenerating codes (SRCs) \cite{SRC,SRC2} are such non-MDS array codes which encode $km$ data symbols to obtain $(n-k)m+n$ coded symbols. The total $n(m+1)$ symbols are stored in $n$ nodes, each node stores $m+1$ symbols. The design idea of SRCs is that we first create $m$ instances of an $(n,k,1)$ MDS code and then design the extra $n$ coded symbols by XORing some symbols of the obtained $nm$ symbols in the first step. We denote the SRC by $(n,k,m)$-SRC.  
In this paper, we propose {\em generalized simple regenerating code} (GRSC) that can not only support much more parameters but also have larger fault tolerance.

\subsection{Examples}
In the following, we present two examples of $(n,k)=(18,16)$ to illustrate our main idea.

Fig. \ref{fig:1.(a)} shows the codeword of $(n=18,k=16,m=2)$-SRC, the $m+1=3$ symbols in the same row are stored in a node. In Fig. \ref{fig:1.(a)}, $(x_{0,i},x_{1,i},\cdots,x_{17,i})^T$ is a codeword of an $(18,16,1)$ MDS code for $i\in\{0,1\}$, where the first $k=16$ symbols are data symbols and the last two symbols are coded symbols.
We claim that $(n=18,k=16,m=2)$-SRC can recover any $n-k+1=3$ erased nodes.
Suppose that nodes 0-2 are erased (the nine symbols with gray part in Fig. \ref{fig:1.(a)}). First, we can download $x_{3,0}$ and $x_{3,0}+x_{2,1}$ to recover the symbol $x_{2,1}$. Then, we can obtain $k=16$ symbols in the second column and thus recover the erased two symbols $x_{0,1},x_{1,1}$. Together with the recovered symbol $x_{1,1}$, we can recover the symbol $x_{2,0}$ by downloading  $x_{2,0}+x_{1,1}$. Next, we can recover the two erased symbols $x_{0,0},x_{1,0}$ in the first column by downloading the other $k=16$ symbols. Finally, we can recover the erased three coded symbols by downloading the corresponding six symbols. Similarly, we can show that $(n=18,k=16,m=2)$-SRC can recover any $n-k+1=3$ erased nodes (refer to Theorem \ref{th4} and Corollary \ref{col5} in Section \ref{sec:2.4}).

Fig. \ref{fig:1.(b)} shows another example of GSRCs with $(n,k,m,a)=(18,16,4,2)$, the $m+a=6$ symbols in the same row are stored in a node. In Fig. \ref{fig:1.(b)}, the $n=18$ symbols in each of the first four columns are codeword of an $(18,16,1)$ MDS code, where the first $k=16$ symbols in each of the first four columns are data symbols and the last two symbols are coded symbols. The symbols in the last two columns are linear combinations of some symbols in the first four columns, where $\alpha$ is the primitive element in $\mathbb{F}_{q}$ and $q>18$. We can see that the two codes in Fig. \ref{fig:1} have the same storage overhead 1.69 (defined as the ratio of the total number of symbols to the total number of data symbols). 
	
		\begin{figure}[htpb]
		\centering
		\subfigure[$(n=18,k=16,m=2)$-SRC.]{           
			\includegraphics[width=3.6cm]{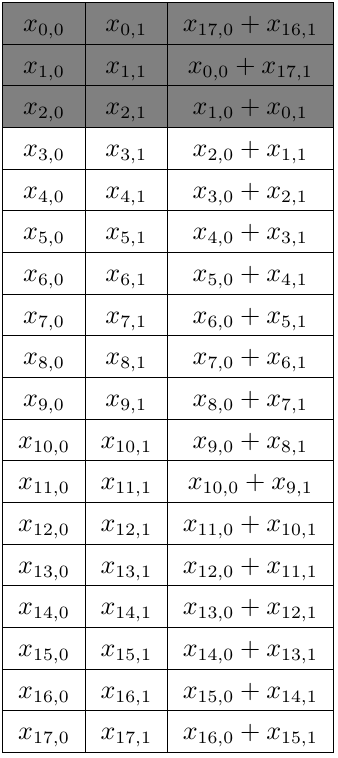}
			\label{fig:1.(a)}}
		\hspace{0in}
		\subfigure[$(n=18,k=16,m=4,a=2)$-GSRC.]{
			\includegraphics[width=11.5cm]{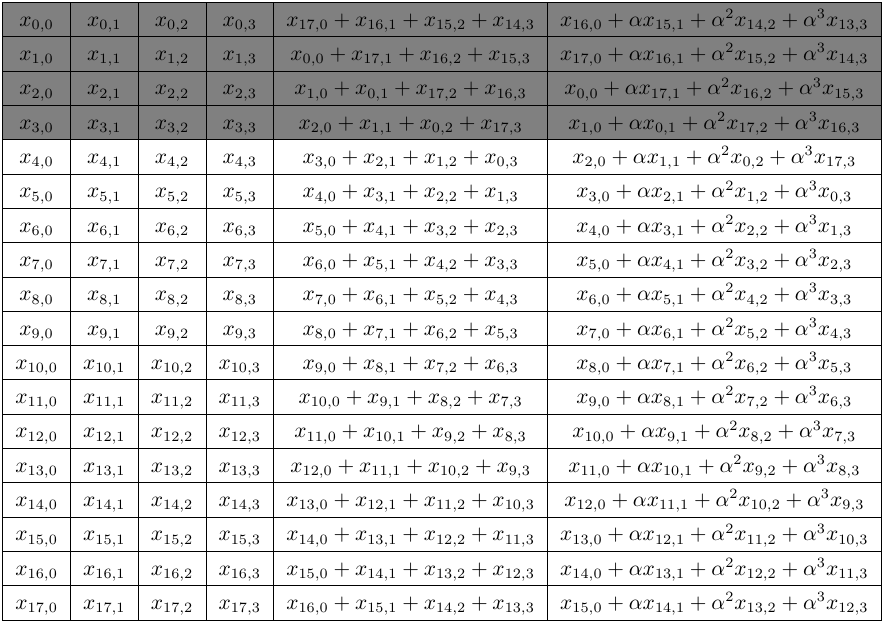}
			\label{fig:1.(b)}}
		\\
		
		\caption{Examples of $(n=18,k=16,m=2)$-SRC (\ref{fig:1.(a)}) and $(n=18,k=16,m=4,a=2)$-GSRC \ref{fig:1.(b)}.}
		\label{fig:1}
	\end{figure}
	
We claim that $(n=18,k=16,m=4,a=2)$-GSRC can recover any $n-k+a=4$ erased nodes. Suppose that nodes 0-3 are erased (the symbols with gray part in Fig. \ref{fig:1.(b)}). First, we can download $x_{4,2},x_{5,1},x_{6,0}$ and $x_{3,3}+x_{4,2}+x_{5,1}+x_{6,0}$ to recover the erased symbol $x_{3,3}$. Then, we download $x_{5,0},x_{4,1}$ and $x_{5,0}+x_{4,1}+x_{3,2}+x_{2,3}, x_{5,0}+\alpha x_{4,1}+\alpha^2x_{3,2}+\alpha^3x_{2,3}$ to recover $x_{3,2},x_{2,3}$ by
	\begin{equation}
		\begin{pmatrix}
			x_{3,2}\\ 
			x_{2,3}
		\end{pmatrix}=\begin{pmatrix}
			1 & 1\\ 
			\alpha^2 & \alpha^3
		\end{pmatrix}^{-1}\begin{pmatrix}\begin{pmatrix}
			x_{5,0}+x_{4,1}+x_{3,2}+x_{2,3}\\ 
			x_{5,0}+\alpha x_{4,1}+\alpha^2x_{3,2}+\alpha^3x_{2,3}
		\end{pmatrix}-\begin{pmatrix}
			1 & 1\\ 
			1 & \alpha
		\end{pmatrix}\begin{pmatrix}
			x_{5,0}\\ 
			x_{4,1}
		\end{pmatrix}\end{pmatrix}.\nonumber
	\end{equation}
Next, we can recover the two erased symbols $x_{0,3},x_{1,3}$ by accessing the other $k=16$ symbols in the fourth column. Together with $x_{1,3}$, we can download $x_{4,0}$ and $x_{4,0}+x_{3,1}+x_{2,2}+x_{1,3}, x_{4,0}+\alpha x_{3,1}+\alpha^2x_{2,2}+\alpha^3x_{1,3}$ to recover $x_{3,1}, x_{2,2}$ by 
		\begin{equation}
		\begin{pmatrix}
			x_{3,1}\\ 
			x_{2,2}
		\end{pmatrix}=\begin{pmatrix}
			1 & 1\\ 
			\alpha & \alpha^2
		\end{pmatrix}^{-1}\begin{pmatrix}\begin{pmatrix}
			x_{4,0}+x_{3,1}+x_{2,2}+x_{1,3}\\ 
			x_{4,0}+\alpha x_{3,1}+\alpha^2x_{2,2}+\alpha^3x_{1,3}
		\end{pmatrix}-\begin{pmatrix}
			1 & 1\\ 
			1 & \alpha^3
		\end{pmatrix}\begin{pmatrix}
			x_{4,0}\\ 
			x_{1,3}
		\end{pmatrix}\end{pmatrix}.\nonumber
	\end{equation}
We can recover the two erased symbols $x_{0,2},x_{1,2}$ by downloading the other surviving $k=16$ symbols in the third column. Together with $x_{0,3},x_{1,2}$, we can download $x_{3,0}+x_{2,1}+x_{1,2}+x_{0,3}$ and $x_{3,0}+\alpha x_{2,1}+\alpha^2x_{1,2}+\alpha^3x_{0,3}$ to recover $x_{3,0}, x_{2,1}$ by 
	\begin{equation}
		\begin{pmatrix}
			x_{3,0}\\ 
			x_{2,1}
		\end{pmatrix}=\begin{pmatrix}
			1 & 1\\ 
			1 & \alpha
		\end{pmatrix}^{-1}\begin{pmatrix}\begin{pmatrix}
			x_{3,0}+x_{2,1}+x_{1,2}+x_{0,3}\\ 
			x_{3,0}+\alpha x_{2,1}+\alpha^2x_{1,2}+\alpha^3x_{0,3}
		\end{pmatrix}-\begin{pmatrix}
			1 & 1\\ 
			\alpha^2 & \alpha^3
		\end{pmatrix}\begin{pmatrix}
			x_{1,2}\\ 
			x_{0,3}
		\end{pmatrix}\end{pmatrix}.\nonumber
	\end{equation}
Similarly, we can recover the erased symbols in the first two columns and finally recover all the erased symbols in the last two columns.
Actually, $(n=18,k=16,m=4,a=2)$-GSRC can recover any $n-k+a=4$ erased nodes (refer to Theorem \ref{th11} in Section \ref{sec:2.5} for the repair method with general parameters).
	
It can be seen from the two examples in Fig. \ref{fig:1} that GSRCs have better fault tolerance than SRCs with the same parameters and storage overhead. In this paper, we present the construction of GSRCs and show that there is a trade-off between sub-packetization and fault tolerance.
	
    \subsection{Contributions}
	Our main contributions are as follows.
	\begin{itemize}
\item [(1)] First, we give construction of GSRCs that can support more parameters when compared to SRCs, specifically, SRCs can be viewed as a special case of GSRCs with $a=1$. We show that GSRCs with $a=1$ can recover any $n-k+1$ erased nodes for most high code-rate parameters, note that the fault tolerance of SRCs is $n-k$ in \cite{SRC}. We also show that GSRCs with $a=1$ can recover most
pattern of $n-k+2$ erased nodes (refer to Corollary \ref{col7}). Moreover, we show that GSRCs with $a=1$ can recover any $2n-2k+1$ erased symbols and 
most pattern of $2n-2k+2$ erased symbols (refer to Corollary \ref{col3}).
\item [(2)] Second, we show that GSRCs can recover any $n-k+a$ erased nodes for most of the parameters (refer to Theorem \ref{th11}), i.e., there is a trade-off between sub-packetization $m+a$ and fault tolerance $n-k+a$.
\item [(3)] Third, we show that our GSRCs have lower repair bandwidth and lower repair locality than that of the existing related codes, such as Locally Recoverable Codes (LRCs) (refer to Theorem \ref{th12}-\ref{th13}).
\end{itemize}

\subsection{Related Works}
There are many constructions of non-MDS codes, such as LRCs \cite{huang2012,2014Locally} with one node storing one symbol, bundles of RAID array codes \cite{Bundle} of which the sub-packetization is no less than two. LRCs divide the data symbols into several groups and obtain some local coded symbols for each group, therefore can recover any single symbol by accessing some other symbols in the same group. Partial maximum distance separable (PMDS) codes \cite{2013PMDS,2016PMDS,2017PMDS,2019PMDS,2020PMDS} are special LRCs which can recover all
erasure patterns that are information theoretically correctable.

Bundles of RAID array codes \cite{Bundle} organize the $km$ data symbols by $k\times m$ array and obtain $n\times (m+1)$ codeword array by first adding $n-k$ local coded symbols for each column and then adding one local coded symbol for each row, where $n=k+1,k+2$.
The $m+1$ symbols in the same row are stored in a node. We show theoretically that our GSRCs have lower repair bandwidth, lower repair locality for single-node erasure and higher fault tolerance than both LRCs and codes in \cite{Bundle} under most high code-rate parameters.

\subsection{Paper Organization}
The rest of the paper is organized as follows. Section \ref{sec:2} presents the construction of GSRCs.
Section \ref{sec:2.4} shows the trade-off between sub-packetization and fault tolerance.
Section \ref{sec:2.5} shows the repair bandwidth for single-node failure of our codes.
Section \ref{sec:3} evaluates the performance for GSRCs and related codes. Section \ref{sec:4} concludes the paper.
	
\section{Generalized Simple Regenerating Codes}
	\label{sec:2}
	In this section, we present the construction of GSRCs that encodes $km$ data symbols into $n\times (m+a)$ array, where $n,k,m,a$ are positive integers with $n>k$ and $n\geq m+a$.
	In this paper, for any integer $x$, we denote $<x>$ as the remainder of $x$ when we divide $x$ by $n$.
	
	\subsection{The Construction}
	\label{sec:2.1}
	
	We represent the $km$ data symbols by an $k\times m$ array and let the symbol in row $j$ and column $i$ be $x_{j,i}$, where $j=0,1,\cdots,k-1$ and $i=0,1,\cdots,m-1$. 
	
	First, we create $n-k$ coded symbols for each column. For $i=0,1,\cdots,m-1$, we encode the $k$ data symbols in column $i$ to obtain $n-k$ coded symbols $x_{k,i},x_{k+1,i},\cdots,x_{n-1,i}$ such that the $n$ symbols $x_{0,i},x_{1,i},\cdots,x_{n-1,i}$ of the obtained $n\times m$ array form a codeword of $(n,k,1)$ MDS codes.
	
	Second, we create $a$ coded symbols for each row. For $j=0,1,\cdots,n-1$, the $a$ coded symbols $p_{j,i}$ with $i=0,1,\cdots,a-1$ are computed as
	\begin{eqnarray}
		p_{j,i}=\sum_{t=0}^{m-1}\alpha^{it}x_{<j-t-1-i>,t},\label{eq.1}
	\end{eqnarray}
	where $\alpha$ is a primitive element in $\mathbb{F}_{q}$ and $q$ is a power of a prime number with $q>n\geq m+a$. We denote the obtained GSRCs as $(n,k,m,a)$-GSRC.
The $n\times (m+a)$ codeword array is shown in Fig. \ref{fig:2}. We label the indices of the $n$ nodes from $0$ to $n-1$ and the indices of the $m+a$ columns from 0 to $m+a-1$. According to Eq. \eqref{eq.1}, for any $j=0,1,\cdots,n-1$, we can see that the $a$ coded symbols $p_{j,0},p_{<j+1>,1},\cdots,p_{<j+a-1>,a-1}$ are linear combinations of the $m$ data symbols $$x_{<j-1>,0},x_{<j-2>,1},\cdots,x_{<j-m>,m-1},$$
i.e.,
\begin{equation}
\begin{bmatrix}
p_{j,0}\\ p_{<j+1>,1}\\ \vdots \\ p_{<j+a-1>,a-1}\\
\end{bmatrix}=\begin{bmatrix}
\alpha^{0}& \alpha^{0}& \cdots & \alpha^{0}\\
\alpha^{0}& \alpha^{1}& \cdots & \alpha^{m-1}\\
\vdots& \vdots& \ddots & \vdots\\
\alpha^{0}& \alpha^{a-1}& \cdots & \alpha^{(a-1)(m-1)}\\
\end{bmatrix}\begin{bmatrix}
x_{<j-1>,0}\\ x_{<j-2>,1}\\ \vdots \\ x_{<j-m>,m-1}\\
\end{bmatrix}.
\label{eq:mds2}
\end{equation}
Without loss of generality, suppose that the $m+a$ symbols in Eq. \eqref{eq:mds2} are a codeword of $(m+a,m,1)$ MDS codes and we can retrieve all the symbols from any $m$ out of the $m+a$ symbols. 
For example, the $m+a$ symbols $$x_{n-1,0},x_{n-2,1},\cdots,x_{n-m,m-1},p_{0,0},p_{1,1},\cdots,p_{a-1,a-1}$$ with bold font in Fig. \ref{fig:2} form a codeword.
When $a=1$, $(n,k,m,a)$-GSRCs are reduced to SRCs in \cite{SRC}. The example in Fig. \ref{fig:1.(b)} is $(18,16,4,2)$-GSRC.
	
	\begin{figure}[htpb]
		\centering
		\includegraphics[width=0.9\linewidth]{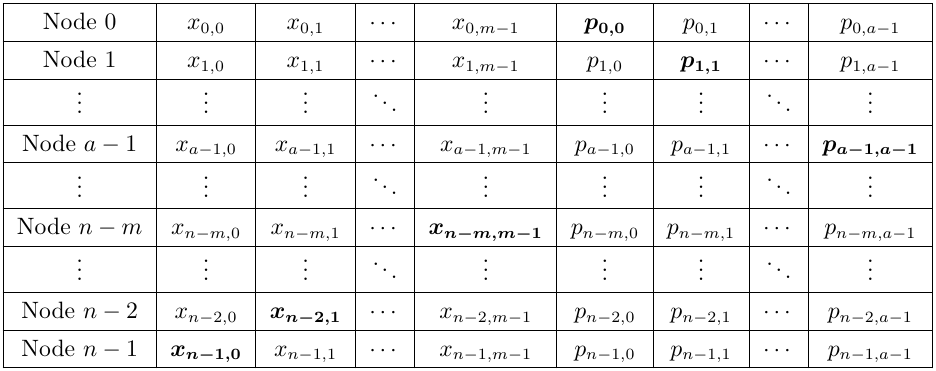}
		\caption{The codeword of $(n,k,m,a)$-GSRC. The $m+a$ symbols with bold font in the figure is a codeword of $(m+a,m,1)$ MDS code.}
		\label{fig:2}
	\end{figure}
	
	{\em Remark:} For $i\in\{0,1,\cdots,m-1\}$ and $j\in\{0,1,\cdots,a-1\}$, denote by $\mathbf{X}_i=(x_{0,i},x_{1,i},\cdots,x_{n-1,i})^T$ and $\mathbf{P}_j=(p_{0,j},p_{1,j},\cdots,p_{n-1,j})^T$. Let $r:=n-k$. Recall that $\mathbf{X}_i$ is a codeword of $(n,k,1)$ MDS code, where $i=0,1,\cdots,m-1$. Suppose that the $r\times n$ parity check matrix of the $(n,k,1)$ MDS code is $H_{r,n}$, we have $H_{r,n}\cdot\mathbf{X}_i=0$, where 
		\begin{eqnarray}
		H_{r,n}=\begin{pmatrix}
			1 & 1 & \cdots & 1\\ 
			1 & \alpha & \cdots & \alpha^{n-1}\\ 
			\vdots & \vdots & \ddots & \vdots\\ 
			1 & \alpha^{r-1} & \cdots &\alpha^{(r-1)(n-1)} 
		\end{pmatrix}.\label{eq.2}
	\end{eqnarray}

	\subsection{Repair Method of Multi-Symbol Erasures for $a=1$}
	\label{sec:2.3}
When $a=1$, by Eq. \eqref{eq.1}, the coded symbol $p_{j,0}$ is
		\begin{eqnarray}
			p_{j,0}=\sum_{t=0}^{m-1}x_{<j-t-1>,t},j\in\{0,1,\cdots,n-1\}.\label{eq.3}
		\end{eqnarray}
We define {\em coded group} $P_{j,0}$ as a set containing the coded symbol $p_{j,0}$ and the $m$ data symbols in Eq. \eqref{eq.3}.
We can repair any one erased symbol in coded group $P_{j,0}$ by downloading the other $m$ symbols.		
 
 We first show that we can repair any $2r+1$ erased symbols.
 
	\begin{theorem}
		\label{th1}
		In $(n,k,m,a=1)$-GSRC, we can repair any $2r+1$ erased symbols.
	\end{theorem}
	\begin{proof}
Suppose that the number of erased symbols in column $i$ is $t_i$, where $0\leq t_i\leq \max (n,2r+1)$ and $i\in\{0,1,\cdots,m\}$ such that $\sum_{i=0}^{m}t_i=2r+1$. 

\textbf{Case 1.}
If $t_i\leq r$ for all $i\in\{0,1,\cdots,m-1\}$, then we can directly repair the erased $\sum_{i=0}^{m-1}t_i$ symbols in the first $m$ columns, since the symbols in columns $i$ form a codeword of $(n,k,1)$ MDS code. We can repair the erased $t_{m}$ coded symbols by Eq. \eqref{eq.3}.
		
\textbf{Case 2.}
If there is a certain $\ell\in\{0,1,\cdots,m-1\}$ such that the number of erased symbols in column $\ell$ is larger than $r$. Suppose that $t_{\ell}=r+a'$, where $1\leq a'\leq r+1$. We have $\sum_{i=0,i\neq \ell}^{m}t_i=r+1-a'$. 

Recall that if one coded group contains only one erased symbol, then we can repair the erased symbol by downloading the other $m$ symbols in the coded group. Otherwise, if one coded group contains more than one erased symbols, we can't repair these erased symbols within the symbols in the coded group. We claim that we can repair at least one erased symbol in column $\ell$ by downloading the other $m$ symbols in the coded group which contains the erased symbol. Suppose that we can't repair any erased symbol in column $\ell$ by the above repair method, then each of the $r+a'$ coded groups which contain the $r+a'$ erased symbols in column $\ell$ contains at least two erased symbols. The erased symbols should be no less than $2r+2a'\geq 2r+2$, which contradicts to that the total number of erased symbols is $2r+1$.
Therefore, we can repair 
at least one erased symbol in column $\ell$ by the coded group and the other erased symbols similarly.
\end{proof}

When the number of erased symbols is $2r+2$, we show in Theorem \ref{th2} and Corollary \ref{col3} that we can recover most of the $2r+2$ erased patterns and some other patterns can't be recovered. 

\begin{theorem}
\label{th2}
In $(n,k,m,a=1)$-GSRC, we can recover the erased $2r+2$ symbols, except that the erased $2r+2$ symbols belong to $r+1$ coded groups and they are located in two columns, each column contains $r+1$ erased symbols.
\end{theorem}
\begin{proof}
See Appendix \ref{app.1}.
\end{proof}

	\begin{corollary}
		\label{col3}
		In $(n,k,m,a=1)$-GSRC, suppose that $2r+2$ symbols are erased, the probability $P$ that we can recover the $2r+2$ erased symbols satisfies 
		\begin{eqnarray}
			P\geq1-\frac{C_{m+1}^{2}\cdot C_{n}^{r+1}}{C_{(m+1)n}^{2r+2}}.\label{eq.4}
		\end{eqnarray}
	\end{corollary}
	\begin{proof}
According to Theorem \ref{th2}, the necessary condition for the $2r+2$ erased symbols to be unrecoverable is that the $2r+2$ erased symbols are located in two columns, each column contains $r+1$ erased symbols and they belong to $r+1$ coded groups.
There are at most $C_{m+1}^{2}\cdot C_{n}^{r+1}$ patterns of the $2r+2$ erased symbols that are unrecoverable. We have Eq. \eqref{eq.4}.
	\end{proof}

Combining Theorem \ref{th1} and Theorem \ref{th2}, we know that the fault tolerance of multi-symbol erasures of $(n,k,m,a=1)$-GSRC is $2r+1$. 
According to Corollary \ref{col3}, we have $P\approx1$ if $r\ll n$, which means that most of $2r+2$ erased patterns can be repaired for high code-rate $(n,k,m,a=1)$-GSRC.

\section{Trade-off Between Sub-packetization And Fault Tolerance}
	\label{sec:2.4}
In this section, we present the repair method of multi-node erasures and show that we can recover any $r+a$ erased nodes in $(n,k,m,a)$-GSRC under most high code-rate parameters. In other words, there is a trade-off between sub-packetization $m+a$ and fault tolerance $r+a$.

Suppose that the $r+a'$ nodes $\{f_i\}_{i=1}^{r+a'}$ are erased, where $1\leq a'$ and $0\leq f_1<f_2<\cdots<f_{r+a'}\leq n-1$. We define $r+a'$ {\em interval values} of the $r+a'$ erased nodes as the differences of two consecutive erased nodes, i.e., let $I_j:=f_j-f_{j-1}-1$ for $j=2,3,\cdots,r+a'$ and $I_{r+a'+1}=n-1+f_1-f_{r+a'}$. We can know that 
	\begin{equation}
		\sum_{j=2}^{r+a'+1}I_j=k-a'.\label{eq.5}
	\end{equation}
We first consider the case of $a=1$. Next theorem shows a sufficient condition to recover $r+a'$ erased nodes when $a=1$.
	
\begin{theorem}
\label{th4}
Suppose that $r+a'$ nodes $\{f_j\}_{j=1}^{r+a'}$ of $(n,k,m,a=1)$-GSRC are erased. 
If there are at least $a'$ elements in the set $\{I_2,I_3,\cdots,I_{r+a'+1}\}$ whose values are no less than $m$, then we can recover the $r+a'$ erased nodes.
\end{theorem}
	
\begin{proof}
Suppose that the $a'$ elements $\{I_{i_j}\}_{j\in\{1,2,\cdots, a'\}}$ are no less than $m$, where $2\leq i_1<i_2<\cdots<i_{a'}\leq r+a'+1$. We have that nodes $\{<f_{i_j-1}+u>\}_{u=1}^{m}$ are not erased, for all $j\in\{1,2,\cdots,a'\}$, according to the definition of interval value.
		
According to Eq. \eqref{eq.3}, for all $j\in\{1,2,\cdots,a'\}$, we have $$p_{<f_{i_j-1}+m>,0}=x_{f_{i_j-1},m-1}+\sum_{t=0}^{m-2}x_{<f_{i_j-1}+m-t-1>,t},$$ 
where the $m$ symbols $p_{<f_{i_j-1}+m>,0},\{x_{<f_{i_j-1}+m-t-1>,t}\}_{t=0}^{m-2}$ are all in the nodes which are not erased. Therefore, we can download the $m$ symbols to recover the erased symbol $x_{f_{i_j-1},m-1}$. After recovering the $a'$ erased symbols $\{x_{f_{i_j-1},m-1}\}_{j=1}^{a'}$ in column $m-1$, we can repair the other $r$ erased symbols  $\{x_{f_j,m-1}\}_{j\in\{1,2,\cdots,r+a'\}\setminus\{i_1-1,i_2-1,\cdots,i_{a'}-1\}}$ in column $m-1$ by the MDS property.

According to Eq. \eqref{eq.3}, for all $j\in\{1,2,\cdots,a'\}$, we can see that $$p_{
<f_{i_j-1}+m-1>,0}=x_{<f_{i_j-1}-1>,m-1}+x_{f_{i_j-1},m-2}+\sum_{t=0}^{m-3}x_{<f_{i_j-1}+m-t-2>,t}.$$ Recall that all the symbols in column $m-1$ have been repaired, the symbols $\{x_{<f_{i_j-1}+m-t-2>,t}\}_{t=0}^{m-3}$ and $p_{<f_{i_j-1}+m-1>,0}$ are all in the nodes that are not erased. Therefore, we can recover the symbol $x_{f_{i_j-1},m-2}$ by the above equation and the other $r$ erased symbols in column $m-2$ by the MDS property. We can similarly recover all the erased symbols in the first $m-1$ columns and then recover all the erased coded symbols in column $m$ by the encoding procedure. 
\end{proof}
	
When $a'=1$, we show that all $r+1$ erased nodes can be recovered under high code-rate condition.
	
\begin{corollary}
\label{col5}
If $k+r>(r+1)m$, then $(n,k,m,a=1)$-GSRC can repair all $r+1$ erased nodes.
\end{corollary}
\begin{proof}
When $a'=1$, we have $\sum_{j=2}^{r+2}I_j=n-r-1=k-1$. If $k+r>(r+1)m$, then 
		\begin{equation}
			(r+1)\cdot \max(I_2,I_3,\cdots,I_{r+2})\geq\sum_{j=2}^{r+2}I_j=k-1>(m-1)(r+1).\nonumber
		\end{equation}
		Thus, we have $\max(I_2,I_3,\cdots,I_{r+2})\geq m$, which means that there exists at least one element in the set $\{I_2,I_3,\cdots,I_{r+2}\}$ which is no less less $m$, the $r+1$ erased nodes can be recovered according to Theorem \ref{th4}.
	\end{proof}
	
The next theorem shows that the fault tolerance of $(n,k,m,a=1)$-GSRC is $r+1$ under the specific condition.
	
\begin{theorem}
\label{th6}
When $k+r>(r+1)m$, the fault tolerance of $(n,k,m,a=1)$-GSRC is $r+1$.
\end{theorem}
\begin{proof}
According to Corollary \ref{col5}, when $k+r>(r+1)m$, $(n,k,m,a=1)$-GSRC can recover any $r+1$ erased nodes. We only need to show that we can't recover some patterns of $r+2$ erased nodes. 

We show that we can't recover the last $r+2$ nodes $\{i\}_{i=n-r-2}^{n-1}$. Recall that the $m+1$ symbols in the last $r$ nodes are all coded symbols. We only need to show that we can't recover the $2m$ data symbols in nodes $k-1$ and $k-2$. Note that only the $m$ coded surviving symbols in nodes $\{i\}_{i=0}^{m-1}$, namely $\{p_{i,0}\}_{i=0}^{m-1}$, are linear combinations of the erased $2m$ data symbols. Since $m<2m$, we can't recover the erased data symbols and the theorem is proved.
\end{proof}

	\begin{corollary}
		\label{col7}
In $(n,k,m,a=1)$-GSRC, suppose that $r+2$ nodes $\{f_i\}_{i=1}^{r+2}$ are erased, the probability $P$ that we can recover the $r+2$ erased nodes satisfies
		\begin{eqnarray}
			P\geq1-\frac{(r+3)!\cdot\lceil \frac{k-1}{2} \rceil\cdot C_{m+r}^{r+1}}{C_{n}^{r+2}}.\nonumber
		\end{eqnarray}
	\end{corollary}
\begin{proof}
Suppose that $r+2$ nodes $\{f_i\}_{i=1}^{r+2}$ are erased. Let $y_1=f_1$ and $y_2=n-1-f_{r+2}$, we have $I_{r+3}=y_1+y_2$ and	
\begin{eqnarray}
\sum_{j=2}^{r+3}I_{j}=(\sum_{j=2}^{r+2}I_{j})+y_1+y_2=k-2,\label{eq.6}
		\end{eqnarray}
according to Eq. \eqref{eq.5}.
		
Note that the number of elements in $\{I_2,I_3,\cdots,I_{r+2},y_{1},y_2\}$ whose values are no less than $m$ is at most one more than the number of elements in $\{I_2,I_3,\cdots,I_{r+2},I_{r+3}\}$ whose values are no less than $m$. If the number of elements in $\{I_2,I_3,\cdots,I_{r+2},y_{1},y_2\}$ whose values are no less than $m$ is no less than three, then we can recover the $r+2$ erased nodes by Theorem \ref{th4}. 
Note that the number of patterns of the $r+2$ erased nodes is equal to the number of non-negative solutions $\{I_2,I_3,\cdots,I_{r+2},y_{1},y_2\}$ of Eq. \eqref{eq.6}. The total number of patterns of the $r+2$ erased nodes is $C_{n}^{r+2}$, we need to calculate the number of non-negative solutions $\{I_2,I_3,\cdots,I_{r+2},y_{1},y_2\}$ of Eq.~\eqref{eq.6} such that the erased $r+2$ nodes may can't be recovered, i.e., the number of non-negative solutions such that the number of elements in $\{I_2,I_3,\cdots,I_{r+2},y_{1},y_2\}$ with at most two values no less than $m$.

Without loss of generality, suppose that the $r+3$ elements $I_2,I_3,\cdots,I_{r+2},y_{1},y_2$ are in increasing order, i.e., $I_2\leq I_3\leq \cdots\leq I_{r+2}\leq y_1\leq y_2$. If $I_{r+2}\geq m$, then $y_2\geq y_1\geq I_{r+2}\geq m$ and we can recover the $r+2$ erased nodes. 
	
Consider that $I_{r+2}\leq m-1$. The number of the non-negative solutions $\{I_2,I_3,\cdots,I_{r+2}\}$ satisfying $0\leq I_2\leq I_3\leq \cdots\leq I_{r+2}\leq m-1$ is equal to the number of solutions $\{I_2,I_3,\cdots,I_{r+2}\}$ satisfying $0< I_2< I_3< \cdots < I_{r+2}< m+r+1$ which is $C_{m+r}^{r+1}$. 
Given the $r+1$ values $\{I_2,I_3,\cdots,I_{r+2}\}$, the number of solutions $\{y_1,y_2\}$ satisfying $y_1+y_2=k-2-(\sum_{j=2}^{r+2}I_{j})$ and $I_{r+2}\leq y_1\leq y_2$ is less than the number of solutions $\{y_1,y_2\}$ such that $y_1+y_2=k-2$ and $0\leq y_1\leq y_2$,
which is upper bounded by $\lceil \frac{k-1}{2} \rceil$.
Therefore, the total number of the non-negative solutions $\{I_2,I_3,\cdots,I_{r+2},y_{1},y_{2}\}$ with arbitrary order and at most two values no less than $m$ is at most $(r+3)!\cdot\lceil \frac{k-1}{2} \rceil\cdot C_{m+r}^{r+1}$ and the result is proved.
	\end{proof}
	
By Corollary \ref{col7}, $(n,k,m,a=1)$-GSRC can recover most patterns of $r+2$ erased nodes, when $\max (r,m)\ll n$.	
In the following, we consider the fault tolerance for general parameter $a\geq 1$. 	

\begin{theorem}
\label{th11}
When $k+r>(r+a)\cdot \max(m,a-1)$, the fault tolerance of $(n,k,m,a)$-GSRC is $r+a$.
\end{theorem}
\begin{proof}
We prove this theorem by mathematical induction for $a\geq 1$. When $a=1$, the result is true by Theorem \ref{th6}. 

Suppose that when $k+r>(r+N)\cdot \max(m,N-1)$, the fault tolerance of $(n,k,m,N)$-GSRC is $r+N$, where $N$ is a positive integer. We will show that the fault tolerance of $(n,k,m,N+1)$-GSRC is $r+N+1$ under the condition of $k+r>(r+N+1)\cdot \max(m,N)$.

When $a=N+1\geq2$, suppose that $r+a=r+N+1$ nodes $\{f_j\}_{j=1}^{r+a}$ are erased. We have that
		\begin{equation}
			\max(I_2,I_3,\cdots,I_{r+a+1})\geq\frac{1}{r+a}\cdot\sum_{j=2}^{r+a+1}I_j=\frac{k-a}{r+a}>\max(m,a-1)-1,\nonumber
		\end{equation}
where the above equation comes from Eq. \eqref{eq.5} and the last inequality comes form that $k+r>(r+a)\cdot \max(m,a-1)$.
Therefore, at least one element in $\{I_2,I_3,\cdots,I_{r+a+1}\}$ is no less than $\max(m,a-1)$. Without loss of generality, suppose that $I_{a+2}\geq \max(m,a-1)$ and $0\leq f_1<f_2<\cdots<f_{r+a}\leq n-1$. Then the nodes $\{f_{a+1}+u\}_{u=1}^{\max(m,a-1)}$ are not erased. 

For $i\in\{1,2,\cdots,r+a\}$, let $t_i=f_i-f_1$, then $0=t_1<t_2<\cdots<t_{r+a}<n$.
In the following, we first repair the first $m$ symbols $\{x_{f_{a+1},\delta}\}_{\delta=0}^{m-1}$ in the erased node $f_{a+1}$, by considering two cases: $m\leq a$ and $m>a$, and then repair all the other erased symbols.

\textbf{Repair the first $m$ symbols of node $f_{a+1}$.} Consider the first case $m\leq a$. For each $\delta\in\{0,1,\cdots,m-1\}$, the $m+a$ symbols 
\begin{equation}
\{x_{f_{a+1}+\delta-u,u}\}_{u=0}^{m-1}\cup\{p_{f_{a+1}+\delta+t+1,t}\}_{t=0}^{a-1}
\label{eq:ma_sym}
\end{equation} 
are a codeword of $(m+a,m)$ MDS codes according to Eq. \eqref{eq:mds2}.
If any $m$ symbols in Eq. \eqref{eq:ma_sym} are known, then we can obtain the other symbols.
First, the $\delta$ symbols $\{x_{f_{a+1}+\delta-u,u}\}_{u=0}^{\delta-1}$ are in surviving nodes and are known. Second, for any $ 0\leq \delta\leq m-1$, the $m-\delta$ symbols $\{p_{f_{a+1}+\delta+t+1,t}\}_{t=0}^{m-\delta-1}$ are known when $0\leq t\leq m-\delta-1$, since $1\leq\delta+t+1\leq m\leq a$. Therefore, at least $m$ symbols are known and we can repair the other symbols in Eq. \eqref{eq:ma_sym}. 
For each $\delta=0,1,\cdots,m-1$, take $u=\delta$, we have $x_{f_{a+1},\delta}\in \{x_{f_{a+1}+\delta-u,u}\}_{u=0}^{m-1}$ and therefore, we have recovered the first $m$ symbols $\{x_{f_{a+1},\delta}\}_{\delta=0}^{m-1}$ in node $f_{a+1}$.
			
Consider the second case $m>a$. For $\delta\in\{0,1,\cdots,m-1\}$, the $\delta$ symbols $\{x_{f_{a+1}+\delta-u,u}\}_{u=0}^{\delta-1}$ are in surviving nodes and are known.

If $\delta\geq m-a$, then we have $m-\delta-1\leq a-1$. Since  $1<\delta+t+1\leq m$ for all $t\leq m-\delta-1\leq a-1$, the $m-\delta$ symbols $\{p_{f_{a+1}+\delta+t+1,t}\}_{t=0}^{m-\delta-1}$ are known. Therefore, we can repair all the erased symbols in Eq. \eqref{eq:ma_sym}.
			
If $m-t_{a+1}\leq \delta\leq m-a-1$. For $0\leq t\leq a-1< m-\delta-1$, we have $1<\delta+t+1<m$, the $a$ symbols $\{p_{f_{a+1}+\delta+t+1,t}\}_{t=0}^{a-1}$ are known.
On the other hand, we can see that there are at most $a$ erased symbols in $\{x_{f_{a+1}+\delta-u,u}\}_{u=0}^{m-1}$ since $m-t_{a+1}\leq\delta$. Therefore, we can obtain all erased symbols in Eq. \eqref{eq:ma_sym} for $m-t_{a+1}\leq\delta\leq m-a-1$ and $m>a$, according to the $(m+a,m)$ MDS property.

When $0\leq\delta\leq m-1-t_{a+1}$, the number of erased symbols in $\{x_{f_{a+1}+\delta-u,u}\}_{u=0}^{m-1}$ is larger than $a$. We can't repair the erased symbols in Eq. \eqref{eq:ma_sym}. However, we show that all the erased symbols in column $m-1$ can be repaired as follows. Recall that we have repaired all erased symbols in the  $t_{a+1}$ symbols $\{x_{f_{a+1}-j,m-1}\}_{j=0}^{t_{a+1}-1}$ in column $m-1$ when $m-t_{a+1}\leq\delta\leq m-1$. Specifically, we have repaired $a$ erased symbols $\{x_{f_{i},m-1}\}_{i=2}^{a+1}$ in column $m-1$. And at this time, the other $r$ erased symbols in column $m-1$ can be repaired, since the $n$ symbols in column $m-1$ are a codeword of an $(n,k)$ MDS code.

Once $x_{f_{1},m-1}$ is recovered, we only have $a$ erased symbols in $\{x_{f_{a+1}+m-1-t_{a+1}-u,u}\}_{u=m-1-t_{a+1}}^{m-1}$, and we can obtain all the symbols in Eq. \eqref{eq:ma_sym} for $\delta=m-1-t_{a+1}$
and repair the erased symbol $x_{f_{a+1},m-t_{a+1}-1}$ in node $f_{a+1}$. Similarly, we can repair the $m-t_{a+1}$ erased symbols $\{x_{f_{a+1},u}\}_{u=0}^{m-1-t_{a+1}}$ in sequence
by the $(n,k)$ MDS property of the $n$ symbols in each of the first $m$ columns and the $(m+a,m)$ MDS property of the $m+a$ symbols in Eq. \eqref{eq:ma_sym}.

Up to now, we have repaired the first $m$ symbols $\{x_{f_{a+1},\delta}\}_{\delta=0}^{m-1}$ of the erased node $f_{a+1}$. 

\textbf{Repair the other erased symbols of node $f_{a+1}$.}		According to Eq. \eqref{eq.1} with $j=f_{a+1}$, we have that the coded symbol $p_{f_{a+1},i}$ is linear combination of the $m$ symbols $\{x_{f_{a+1}-t-i-1,t}\}_{t=0}^{m-1}$, where $i=0,1,\cdots,a-2$. Let $j=f_{a+1}+a-\ell-1$ and $i=a-1$ in Eq. \eqref{eq.1}, then the coded symbol $p_{f_{a+1}+a-\ell-1,a-1}$ is linear combination of the $m$ symbols $\{x_{f_{a+1}-t-\ell-1,t}\}_{t=0}^{m-1}$, where $\ell=0,1,\cdots,a-2$. Therefore, for each $i=0,1,\cdots,a-2$, both $p_{f_{a+1},i}$ and $p_{f_{a+1}+a-i-1,a-1}$ are linear combinations of the $m$ symbols $\{x_{f_{a+1}-t-i-1,t}\}_{t=0}^{m-1}$.
		
For $0\leq i\leq a-2$, we have $1\leq a-i-1\leq a-1$, then the node $f_{a+1}+a-i-1$ is not erased and we can obtain the symbol $p_{f_{a+1}+a-i-1,a-1}$. By replacing the erased symbols $p_{f_{a+1},i}$ by the symbol $p_{f_{a+1}+a-i-1,a-1}$ for $0\leq i\leq a-2$, we obtain the $n\times (m+a)$ array such that the $m$ data symbols and the first $a-1$ coded symbols in node $f_{a+1}$ are known. Looking at the first $m+a-1$ columns of the array, i.e.,
\begin{align*}
\begin{bmatrix}
x_{0,0}&x_{0,1}&\cdots &x_{0,m-1}&p_{0,0}&p_{0,1}&\cdots &p_{0,a-2}\\
x_{1,0}&x_{1,1}&\cdots &x_{1,m-1}&p_{1,0}&p_{1,1}&\cdots &p_{1,a-2}\\
\vdots&\vdots &\ddots &\vdots &\vdots &\vdots &\ddots &\vdots \\
x_{f_{a+1},0}&x_{f_{a+1},1}&\cdots &x_{f_{a+1},m-1}&p_{f_{a+1}+a-1,a-1}&p_{f_{a+1}+a-2,a-1}&\cdots &p_{f_{a+1}+1,a-1}\\
\vdots&\vdots &\ddots &\vdots &\vdots &\vdots &\ddots &\vdots \\
x_{n-1,0}&x_{n-1,1}&\cdots &x_{n-1,m-1}&p_{n-1,0}&p_{n-1,1}&\cdots &p_{n-1,a-2}\\
\end{bmatrix},
\end{align*}
the symbols in the $r+a-1$ rows $f_1,f_2,\ldots,f_a,f_{a+2},\ldots,f_{r+a}$ are erased and the other symbols are known. It is sufficient to show that we can repair the erased $r+a-1$ rows from the above $n\times (m+a-1)$ array.

Note that the $n$ symbols in each column of the first $m$ columns of the $n\times (m+a-1)$ array are a codeword of $(n,k)$ MDS code and the $m+a-1$ symbols 
\begin{align*}
&\{x_{<j-i-1>,i}\}_{i=0}^{m-1}\cup\{p_{<j+i>,i}\}_{i=0}^{a-2}\setminus \{p_{f_{a+1},<f_{a+1}-j>}\} \cup\{p_{<j+a-1>,a-1}\} \text{ if } f_{a+1}\in \{<j+i>\}_{i=0}^{a-2},\\
&\{x_{<j-i-1>,i}\}_{i=0}^{m-1}\cup\{p_{<j+i>,i}\}_{i=0}^{a-2} \text{ if } f_{a+1}\not\in \{<j+i>\}_{i=0}^{a-2},
\end{align*}
are a codeword of $(m+a-1,m)$ MDS code. According to the assumption that the fault tolerance of $(n,k,m,a-1)$-GSRC is $r+a-1$, we can repair all the $r+a-1$ erased rows.
Therefore, our $(n,k,m,a=N+1)$-GSRC can repair any $r+N+1$ erased nodes, if $k>(r+a)\cdot max(m,a-1)-r$. 

Suppose that both the first $a+1$ nodes and the last $r$ nodes in $(n,k,m,a)$-GSRC are erased, we can always show that we can not repair the erased $r+a+1$ nodes and thus finish the proof. Please refer to the detailed proof in Appendix \ref{app.2}.
	\end{proof}

Continue the example of $(n,k,m,a)=(18,16,2,1)$ in Fig. \ref{fig:1.(a)}. We can check that the condition in Theorem \ref{th11} holds, we can recover any $r+a=3$ erased nodes. However, we can't recover some four erased nodes. Suppose that the first two nodes and the first two nodes are erased. It is sufficient to recover the erased four data symbols $x_{0,0},x_{0,1},x_{1,0},x_{1,1}$, because the other erased symbols are coded symbols. However, there are only two symbols in the surviving nodes which are linear combinations of the four data symbols. It is impossible to repair the erased four data symbols, and therefore can't repair the erased four nodes.	
	
\section{The Repair Bandwidth for Single-node Failure of $(n,k,m,a)$-GSRC}
\label{sec:2.5}
In this section, we analyze the repair bandwidth of single-node erasure for $(n,k,m,a)$-GSRC. 
We define the {\em average repair bandwidth ratio} as the ratio of the average repair bandwidth of all $n$ nodes to the total number of data symbols. We define the {\em repair locality} as the number of nodes contacted in repairing one single-node erasure.
	
\begin{theorem}
\label{th10}
In $(n,k,m,a)$-GSRC, the average repair bandwidth ratio  is $\frac{m+a}{k}$, the repair locality of each node is $\min(2m+a-1,n-1)$.
\end{theorem}
\begin{proof}
Suppose node $f\in\{0,1,\cdots,n-1\}$ is erased.
Recall that the first $m+1$ columns of the $n\times(m+a)$ codeword array is a codeword of $(n,k,m)$-SRC \cite{SRC}.
According to the repair method \cite[Theorem 4]{SRC} of $(n,k,m)$-SRC, we can repair the first $m+1$ symbols in node $f$ by downloading $m(m+1)$ symbols from nodes $\{<f+i>,<f-i>\}_{i=1}^{m}$.
		
The last $a-1$ symbols in node $f$ are $p_{f,i}=\sum_{t=0}^{m-1}\alpha^{it}x_{<f-t-1-i>,t}$, where $i=1,2,\ldots,a-1$. For any $i\in\{1,2,\cdots,a-1\}$ and $t\in\{0,1,\cdots,m-1\}$, we have $2\leq t+i+1\leq m+a-1$ and $<f-t-1-i>\in\{<f-u>\}_{u=2}^{m+a-1}$.
Note that $f\notin\{<f-u>\}_{u=2}^{m+a-1}$. Otherwise, suppose that $f=(f-u)\bmod n$, then $u$ should be a multiple of $n$, which contradicts with $2\leq u\leq m+a-1\leq n-1$. Therefore, we can repair the symbol $p_{f,i}$ by downloading the $m$ symbols $\{x_{<f-t-1-i>,t}\}_{t=0}^{m-1}$, for $i=1,2,\cdots,a-1$.
The repair bandwidth of $(n,k,m,a)$-GSRC is $m(m+1)+m(a-1)=m(m+a)$ symbols and the average repair bandwidth ratio is $\frac{m(m+a)}{mk}=\frac{m+a}{k}$.
		
In our repair method, we repair the symbol $p_{f,i}$ by downloading symbols from nodes $\{<f-t-1-i>\}_{t=0}^{m-1}$, where $i=1,2,\cdots,a-1$. Therefore, the erased node is repaired by downloading symbols from the following nodes
\begin{align*}
&\{<f+i>,<f-i>\}_{i=1}^{m}\bigcup\cup_{i\in\{1,\cdots,a-1\}}(\{<f-t-1-i>\}_{t=0}^{m-1})\\
=&\{<f+i>\}_{i=1}^{m}\cup\{<f-i>\}_{i=1}^{m+a-1},
\end{align*}
and the repair locality is $\min(2m+a-1,n-1)$.
	\end{proof}

Continue the example of $(n,k,m,a)=(18,16,4,2)$ in Fig. \ref{fig:1.(b)}. Suppose that the first node is erased, i.e., $f=0$. According to the repair method in the proof of Theorem \ref{th10}, we can repair the first $m+1=5$ symbols in node $f=0$ by downloading the following $m(m+1)=20$ symbols 
\begin{align*}
&x_{0,0}+x_{17,1}+x_{16,2}+x_{15,3},x_{17,1},x_{16,2},x_{15,3},\\
&x_{1,0}+x_{0,1}+x_{17,2}+x_{16,3},x_{1,0},x_{17,2},x_{16,3},\\
&x_{2,0}+x_{1,1}+x_{0,2}+x_{17,3},x_{2,0},x_{1,1},x_{17,3},\\
&x_{3,0}+x_{2,1}+x_{1,2}+x_{0,3},x_{3,0},x_{2,1},x_{1,2},\\
&x_{17,0},x_{16,1},x_{15,2},x_{14,3},
\end{align*}
from nodes $\{1,2,3,4,14,15,16,17\}$. We can repair the last erased symbol $x_{16,0}+\alpha x_{15,1}+\alpha^2 x_{14,2}+\alpha^3 x_{13,3}$ by downloading the four symbols $x_{16,0}, x_{15,1}, x_{14,2}, x_{13,3}$ from nodes $\{13,14,15,16\}$. Therefore, the repair bandwidth of node $f=0$ is 24 symbols and the repair locality is $\min(2m+a-1,n-1)=9$.

\section{Comparisons}
\label{sec:3}
In this section, we evaluate the performance for our GSRCs and the existing related codes, including SRCs, bundles of RAID array codes and typical LRCs. 
	
\subsection{Comparisons with GSRCs and SRCs}
	\label{sec:3.1}
We summarize the performance of our GSRCs and SRCs in Table \ref{tab:1}.
\begin{table}[htpb] 
\centering
\caption{Comparisons with $(n,k,m,a)$-GSRC and $(n,k,m)$-SRC.}
  \resizebox{\linewidth}{!}{
		\begin{tabular}{|c|c|c|}
			\hline %
			& $(n,k,m,a)$-GSRC & $(n,k,m)$-SRC\\ \hline
			Storage overhead & $\frac{(m+a)n}{mk}$ & $\frac{(m+1)n}{mk}$  \\ \hline
			Sub-packetization & $m+a$ & $m+1$  \\ \hline
			Average repair bandwidth ratio &$\frac{m+a}{k}$  &$\frac{m+1}{k}$	\\ \hline
			Fault tolerance& $n-k+a$ sub. $n>(r+a)\cdot \max(m,a-1)$ & $n-k+1$ sub. $n>(r+1)m$\\ \hline
			Repair locality & $\min(2m+a-1,n-1)$ & $\min(2m,n-1)$\\ \hline
		\end{tabular}
		\label{tab:1}}
	\end{table}
	
From the results in Table. \ref{tab:1}, we can see that the fault tolerance, repair bandwidth, repair locality, sub-packetization and storage overhead of $(n,k,m,a)$-GSRC increases linearly with $a$.

Consider $a$ instances of $(n,k,m)$-SRC and one instance of $(n,k,am,a)$-GSRC such that both codes have the same storage overhead and sub-packetizaiton, while the fault tolerance of $(n,k,am,a)$-GSRC is $n-k+a$, which is larger than that of $(n,k,m)$-SRC. Our GSRCs have better tradeoff between storage overhead and fault tolerance than the existing SRCs.

\subsection{Comparisons with GSRCs and Bundles of RAID Array Codes}
\label{sec:3.15}
Bundles of RAID array codes \cite{Bundle} encode $km$ data symbols into an $n\times (m+1)$ array, the $m+1$ symbols are stored in one node, where  $n=k+1,k+2$ and $m\ll n$.

\begin{table}[htpb]
\centering
\caption{Comparisons with $(n,k,m,a=1)$-GSRC and bundles of RAID array codes.}
    \resizebox{\linewidth}{!}{
	\begin{tabular}{|c|c|c|c|c|}
		\hline
&Codes \cite{Bundle} with $n=k+1$ &GSRCs with $n=k+1$ &Codes \cite{Bundle} with $n=k+2$&GSRCs with $n=k+2$\\
		\hline
        Storage overhead&$\frac{(m+1)n}{m(n-1)}$&$\frac{(m+1)n}{m(n-1)}$&$\frac{(m+1)n}{m(n-2)}$&$\frac{(m+1)n}{m(n-2)}$\\
		\hline
		Repair locality&$n-1$&$2m$&$n-1$&$2m$\\
		\hline
		Average repair bandwidth ratio&1&$\frac{m+1}{n-1}$&1&$\frac{m+1}{n-2}$\\
		\hline
		Fault tolerance &1&2&2&3\\
		\hline
		
\end{tabular}\label{tab:1.5}}
\end{table}

Table \ref{tab:1.5} summarizes the performance of our GSRCs and bundles of RAID array codes, where the sub-packetization of both codes is $m+1$. 
According to Table \ref{tab:1.5}, we can observe that our GSRCs have better performance compared with bundles of RAID array codes, in terms of repair bandwidth, repair locality and fault tolerance.

\subsection{Comparisons with GSRCs and LRCs}
\label{sec:3.2}
In the following, we evaluate the performance of our codes, optimal-LRCs \cite{2014Locally}, and locally MSR PMDS codes \cite{2021PMDS}, in terms of repair bandwidth, storage overhead, repair locality and fault tolerance. 
	
We review the construction of optimal-LRCs and locally MSR PMDS codes.
We have $k'+r'$ nodes and each node stores $\alpha$ symbols, where $k',r',\alpha$ are all positive integer. We need to encode $\alpha(k'g-s)$ data symbols to obtain $(k'+r')\alpha$ symbols that are stored in $k'+r'$ nodes, where $s$ is positive integer.
We first encode all $\alpha(k'g-s)$ data symbols to obtain $s\alpha$ global coded symbols, divides the $k'g\alpha$ symbols into $g$ groups each group with $k'\alpha$ symbols. We then encode the $k'\alpha$ symbols in each group to obtain $(k'+r')\alpha$ symbols by employing an $(k'+r',k',\alpha)$ MDS array code, where the obtained $(k'+r')\alpha$ symbols are stored in $k'+r'$ nodes.
Optimal-LRCs are the above codes with $\alpha=1$ and locally MSR PMDS codes are the above codes by choosing the $(k'+r',k',\alpha)$ MDS array code to be an $(k'+r',k',\alpha)$ MSR code.
In locally MSR PMDS codes, any single-node erasure can be locally repaired by the $(k'+r',k',\alpha)$ MSR code, the repair bandwidth is $\frac{(k'+r'-1)\alpha}{r'}$ symbols and the repair locality is $k'+r'-1$. Recall that the fault tolerance of an $(g,k',r',s,\alpha)$ optimal-LRC is $r'+s$.

The next theorem shows that our codes have less storage overhead and less repair bandwidth than that of optimal-LRCs under the same fault tolerance.

\begin{theorem}
		\label{th12}
Suppose that $g\cdot\frac{k'+1}{k'}>s+1$. Our $(n=g(k'+1),k=gk'-s+g,m=k',a=1)$-GSRC have the same fault tolerance, however have strictly less storage overhead and less repair bandwidth,
compared with $(g,k',r'=1,s,\alpha=1)$ optimal-LRC codes.
\end{theorem}
\begin{proof}
By assumption, we have that $k'(s+1)<g(k'+1)$, $n=g(k'+1)$, $k=gk'-s+g$ and $m=k'$, then we can obtain that
\begin{align*}
(r+1)m-r=k'(s+1)-s<g(k'+1)-s=k.
\end{align*}
By Theorem \ref{th11}, the fault tolerance of our codes is $r+1=s+1$, which is equal to the fault tolerance of $(g,k',r'=1,s,\alpha=1)$ optimal-LRC code.

When $m=k'$, the storage overhead of our codes is $\frac{(m+a)\cdot n}{mk}=\frac{(k'+1)n}{k'(gk'-s+g)}$,
and the storage overhead of $(g,k',r'=1,s,\alpha=1)$ optimal-LRC code is $\frac{n}{gk'-s}$. We have that
		\begin{align*}
			\frac{(k'+1)n}{k'(gk'-s+g)}<\frac{n}{gk'-s}\Leftrightarrow (gk'-s)(k'+1)<k'(gk'-s+g)\Leftrightarrow 0<s.
		\end{align*}
Therefore, our codes have strictly less storage overhead than that of		$(g,k',r'=1,s,\alpha=1)$ optimal-LRC code.

By Theorem \ref{th10}, the average repair bandwidth ratio of our codes is $\frac{m+a}{k}$. Recall that the average repair bandwidth of 
$(g,k',r'=1,s,\alpha=1)$ optimal-LRC code is $\frac{(k'+r'-1)\alpha}{r'}\cdot\frac{1}{(gk'-s)\alpha}=\frac{k'}{gk'-s}$. We have that
		\begin{align*}
			\frac{m+a}{k}<\frac{k'}{gk'-s}\Leftrightarrow\frac{k'+1}{gk'-s+g}<\frac{k'}{gk'-s}\Leftrightarrow (gk'-s)(k'+1)<k'(gk'-s+g)\Leftrightarrow 0<s.
		\end{align*}
Therefore, our codes have strictly less repair bandwidth than that of $(g,k',r'=1,s,\alpha=1)$ optimal-LRC code. 
		
	\end{proof}
The next theorem shows that our codes have better performance than that of locally MSR PMDS codes \cite{2021PMDS}, in terms of sub-packetization level, fault tolerance and repair locality.
\begin{theorem}
\label{th13}
Suppose that  $r'\geq2$, $1\leq s\leq k'$, $g\geq2$ and $\alpha=r'^{k'+r'-1}$ in locally MSR PMDS codes	\cite{2021PMDS}.
Let $n=g(k'+r')$, $k=n-r'-s$ and $m=\frac{a(gk'-s)}{gr'-r'}$ in our codes, where $2a\leq r'$. If $g\geq \max(\frac{2k'}{r'}+1,\frac{r'+s+a}{2})$, then our codes have better performance as follows.	\begin{enumerate}
\item[1)] Our codes have the same storage overhead as locally MSR PMDS codes.
\item [2)] The sub-packetization of our codes is lower than $a(\frac{3}{2}+\frac{k'}{r'})$.
\item[3)] The fault tolerance of our codes 
and locally MSR PMDS codes are $r'+s+a$ and $r'+s$, respectively.
\item[4)] Our codes have smaller repair locality than that of of locally MSR PMDS codes.
\end{enumerate}
\end{theorem}
\begin{proof}
\begin{enumerate}
\item[1)] The storage overhead of our codes is $$\frac{(m+a)n}{mk}=\frac{n}{k}(1+\frac{a}{m})=\frac{n}{k}(1+\frac{gr'-r'}{gk'-s})=\frac{n}{k}\cdot\frac{gk'+gr'-r'-s}{gk'-s}=\frac{n}{gk'-s},$$ 
which is the same as that of locally MSR PMDS codes.
\item [2)] Since  $g\geq\frac{2k'}{r'}+1$, i.e., $(g-1)r'\geq 2k'$, we have
			\begin{align*}
				m=\frac{a(gk'-s)}{gr'-r'}<\frac{agk'}{gr'-r'}=a(\frac{k'}{r'}+\frac{k'}{(g-1)r'})\leq a(\frac{k'}{r'}+\frac{1}{2}).
			\end{align*}
We can obtain that  $m+a<a(\frac{3}{2}+\frac{k'}{r'})$.
\item[3)] Since $g\geq\frac{r'+s+a}{2}+1=\frac{r+a}{2}+1$, i.e., $2(g-1)\geq r+a$ and $2a\leq r'$ by the assumption, we have	\begin{align*}
k-((r+a)m-r)&=(k+r)-(r+a)m
        \nonumber\\
        &=g(k'+r')-(r+a)\cdot\frac{a(gk'-s)}{(g-1)r'}\nonumber
				\\ 
            &\geq g(k'+r')-(r+a)\cdot\frac{gk'}{2(g-1)}
            \nonumber\\
            &=\frac{g}{2(g-1)}\cdot(2(g-1)(k'+r')-(r+a)k')\\
&\geq \frac{gr'(r+a)}{2(g-1)}>0.\nonumber
			\end{align*}
We can also obtain that	\begin{align*}
k+r-((r+a)(a-1))&=g(k'+r')-(r+a)(a-1)\nonumber\\
                &> \frac{r+a}{2}\cdot(k'+r')-(r+a)(a-1)
				\nonumber\\
            &=(r+a)(\frac{k'}{2}+1+(\frac{r'}{2}-a))\\
            &\geq (r+a)\cdot(\frac{k'}{2}+1)>0.\nonumber
			\end{align*}
Therefore, we have $k+r>(r+a)\cdot \max(m,a-1)$, the fault tolerance of our codes is $n-k+a=r'+s+a$ according to Theorem \ref{th11}.
The fault tolerance of locally MSR PMDS codes	\cite{2021PMDS} is $r'+s$.
\item[4)] 
By Theorem \ref{th10}, the repair locality of our codes is $2m+a-1$. The repair locality of locally MSR PMDS codes	\cite{2021PMDS} is $k'+r'-1$. We can compute that
\begin{align*}
(2m+a-1)-(k'+r'-1)&=2\cdot\frac{a(gk'-s)}{(g-1)r'}+a-1-(k'+r'-1)\\
&\leq\frac{gk'-s}{g-1}+\frac{r'}{2}-1-(k'+r'-1)\\
&<\frac{gk'}{g-1}+\frac{r'}{2}-1-(k'+r'-1)\\
&=\frac{2k'-(g-1)r'}{2(g-1)}\\
&\leq 0,		\end{align*}
where the first equality comes from the assumption that $m=\frac{a(gk'-s)}{gr'-r'}$, the first inequality comes from that $a\leq\frac{r'}{2}$, and the last inequality comes from the assumption that $g\geq \max(\frac{2k'}{r'}+1,\frac{r'+s+a}{2})$.
Therefore, the repair locality of our codes is smaller than that of locally MSR PMDS codes	\cite{2021PMDS}.
		\end{enumerate}
			\end{proof}
	
According to Theorem \ref{th13}, our codes have smaller repair locality, larger fault tolerance and much lower sub-packetization level, compared with locally MSR PMDS codes \cite{2021PMDS}, under most high code-rate parameters. Table \ref{tab:2.5} shows some supported values of parameters $(n,k,m,a)$ of Theorem \ref{th13}.

\begin{table}[htpb]
	\centering
	\caption{Some supported parameters $(n,k,m,a)$ of Theorem \ref{th13}}
		\begin{tabular}{|c|c|c|c|}
			\hline
			&$n$ &$k$ &$m$\\
			\hline
		$a=1$&$96$&$81, 82, \ldots, 93$&$4$\\
			\hline
		$a=2$&$112$&$100,101,\ldots,107$&$7$\\
			\hline
			$a=3$&$300$&$283,284,\ldots,293$&$14$\\
			\hline
			
		\end{tabular}\label{tab:2.5}
\end{table}	
	\section{Conclusion}
	\label{sec:4}
	
In this paper, we present the construction of GSRCs which generalize the SRCs. We deduce the fault tolerance for our GSRCs and show that there is a trade-off between sub-packetization and fault tolerance. We show that our codes have better performance than the existing related codes, such as LRCs, in terms of repair bandwidth, fault tolerance and repair locality. The implementation of our codes in distributed storage systems is one of our future work.

        \begin{appendices}

        \section{Proof of Theorem \ref{th2}}
        \label{app.1}
Suppose that $2r+2$ symbols are erased and the number of erased symbols in column $i$ is $t_i$ for $i\in\{0,1,\cdots,m\}$. We have $\sum_{i=0}^{m}t_i=2r+2$.

If $t_i\leq r$ for all $i\in\{0,1,\cdots,m-1\}$, we can directly repair $\sum_{i=0}^{m-1}t_i$ erased symbols in the first $m$ columns and then repair the erased $t_{m}$ symbols in column $m$ by Eq. \eqref{eq.3}.

If there is a certain $\ell\in\{0,1,\cdots,m-1\}$ with $t_{\ell}>r$. Suppose that $t_{\ell}=r+a'$, where $1\leq a'\leq r+2$. We have $\sum_{i=0,i\neq \ell}^{m}t_i=r+2-a'$. Since $1\leq a'$, we have $r+a'\geq r+2-a'$ and at least $(r+a')-(r+2-a')=2a'-2$ erased symbols in column $\ell$ which are in the coded group such that all the other symbols in this coded group are not erased. Therefore, we can repair these $2a'-2$ erased symbols in column $\ell$.

If $a'\geq2$, we have $2a'-2\geq 2$. After repairing the $2a'-2$ erased symbols in column $\ell$, we can repair the other $$(2r+2)-(2a'-2)=2r+4-2a'<2r+1$$ erased symbols according to Theorem \ref{th1}. 
If $a'=1$, we have $t_{\ell}=r+1$, where $\ell\in\{0,1,\cdots,m-1\}$. If the other $r+1$ erased symbols are not in one column, then there is at least one $\ell'(\neq\ell)\in\{0,1,\cdots,m-1\}$, so that $1\leq t_{\ell'}\leq r$, and these $t_{\ell'}$ symbols can be repaired according to $(n,k)$ MDS property,  then we can repair the other no more than $2r+1$ erased symbols according to Theorem \ref{th1}.
Therefore, we can repair the erased $2r+2$ symbols, except that \textbf{Case I:} the erased $2r+2$ symbols belongs to $r+1$ coded groups and they are located in two columns, each column contains $r+1$ erased symbols.

In the following, we show that we can repair some patterns of \textbf{Case I}, while can't repair the other patterns of \textbf{Case I}. First, we present some notations.
In our $(n,k,m,a=1)$-GSRC, we represent $n\times (m+1)$ symbols by a column vector of length $n(m+1)$,
$$\mathbf{X}:=(\mathbf{X}_0^T,\mathbf{X}_1^T,\cdots,\mathbf{X}_{m-1}^T,\mathbf{P}_0^T)^T,$$
where $\mathbf{X}_i$ denotes the $n$ symbols in column $i$ with $i=0,1,\ldots,m-1$ and $\mathbf{P}_0$ denotes the $n$ symbols in column $m$. 

Given a matrix $H$, we denote the entry in row $j$ and column $i$ as $H(j,i)$. 
According to Eq. \eqref{eq.3}, we have $n$ linear equations
$$H_{0}\cdot\mathbf{X}=0,$$ where $H_{0}$ is an $n\times (m+1)n$ sparse matrix, in which all the entries are 0 except that the following entries are 1,
		\begin{eqnarray}
			&&\forall j\in\{1,2,\cdots,n\}, H_{0}(j,mn+j)=1,\nonumber\\
			&&\forall j\in\{1,2,\cdots,n\},t\in\{1,2,\cdots,m\}, H_{0}(j,<j-1-t>+(t-1)n+1)=1.\nonumber
		\end{eqnarray} 
Together with the $r=n-k$ linear constraints in each of the first $m$ columns, we can obtain $n+rm$ linear equations $H_{n,r,m}\cdot\mathbf{X}=0$, $H_{n,r,m}$ is the following $(n+rm)\times((m+1)n)$ parity-check matrix, 
		\begin{eqnarray}
			H_{n,r,m}=\begin{pmatrix}
				H_{r,n} & \mathbf{0}_{r,n} & \cdots & \mathbf{0}_{r,n} &\mathbf{0}_{r,n}\\
				\mathbf{0}_{r,n} & H_{r,n} & \cdots & \mathbf{0}_{r,n}&\mathbf{0}_{r,n}\\
				\vdots & \vdots& \ddots& \vdots&\vdots\\
				\mathbf{0}_{r,n}&\mathbf{0}_{r,n}& \cdots & H_{r,n} &\mathbf{0}_{r,n}\\
				\hline 
				&&H_{0}&&\\ 
			\end{pmatrix},
   \label{eq:parity}
		\end{eqnarray}
		where $\mathbf{0}_{r,n}$ is the $r\times n$ zero matrix and $H_{r,n}$ is in Eq. \eqref{eq.2}. 

Consider \textbf{Case I}.
Suppose that $r\geq2$ and $2r+2$ erased symbols are in columns $\hat{t}_1$ and $\hat{t}_2$, where $0\leq \hat{t}_1<\hat{t}_2\leq m$. We will show that we can repair the erased symbols if $\hat{t}_2\leq m-1$ under some conditions and can't repair the erased symbols if $\hat{t}_2= m$. 

We first consider that $\hat{t}_2\leq m-1$. Suppose that the $r+1$ erased symbols located in column $\hat{t}_1$ are $\{x_{l_i,\hat{t}_1}\}_{i=1}^{r+1}$, where $0\leq l_1<l_2<\cdots<l_{r+1}\leq n-1$. Since they belong to $r+1$ coded groups, the $r+1$ erased symbols in column $\hat{t}_2$ must be $\{x_{<l_i+\hat{t}_1-\hat{t}_2>,\hat{t}_2}\}_{i=1}^{r+1}$. 
Here we assume that $\{l_{i}'\}_{i=1}^{r+1}$ is in ascending order of $\{<l_i+\hat{t}_1-\hat{t}_2>\}_{i=1}^{r+1}$, i.e., $\{l_{i}'\}_{i=1}^{r+1}=\{<l_i+\hat{t}_1-\hat{t}_2>\}_{i=1}^{r+1}$ and $0\leq l_{1}'<l_{2}'<\cdots<l_{r+1}'\leq n-1$. According to the parity-check matrix in Eq. \eqref{eq:parity}, the erased symbols can be repaired if and only if the rank of the following matrix $M$ is $2r+2$.
		\begin{eqnarray}
	M=		\begin{pmatrix}
				A&\mathbf{0}_{r,r+1}\\
				\mathbf{0}_{r,r+1}&B'\\
				P_{\hat{t}_1}'&P_{\hat{t}_2}'\\
			\end{pmatrix},\nonumber
		\end{eqnarray} where $P_{\hat{t}_1}'$ and $P_{\hat{t}_2}'$ are $(r+1)\times(r+1)$ permutation matrix, 
  $A$ is the $r\times(r+1)$ matrix
		\begin{eqnarray}
A=\begin{pmatrix}
				1 & 1 &\cdots&1\\
				\alpha^{l_1}&\alpha^{l_2}&\cdots&\alpha^{l_{r+1}} \\
				\vdots & \vdots& \ddots&\vdots\\
				\alpha^{(r-1)l_1}&\alpha^{(r-1)l_2}&\cdots&\alpha^{(r-1)l_{r+1}} \\
			\end{pmatrix}\nonumber,
		\end{eqnarray}$B'$ is the $r\times(r+1)$ matrix
		\begin{eqnarray}
B'=			\begin{pmatrix}
				1 & 1 &\cdots&1\\
				\alpha^{l_1'}&\alpha^{l_2'}&\cdots&\alpha^{l_{r+1}'} \\
				\vdots & \vdots& \ddots&\vdots\\
				\alpha^{(r-1)l_1'}&\alpha^{(r-1)l_2'}&\cdots&\alpha^{(r-1)l_{r+1}'} \\
			\end{pmatrix}\nonumber.
		\end{eqnarray}
By swapping some rows and columns of the matrix $M$, we can obtain the following matrix $Q$.
		\begin{eqnarray}
			M\xrightarrow[]{1}\begin{pmatrix}
				A&\mathbf{0}_{r,r+1}\\
				\mathbf{0}_{r,r+1}&B\\
				P_{\hat{t}_1}'&P_{\hat{t}_1}'\\
			\end{pmatrix}\xrightarrow[]{2}Q:=\begin{pmatrix}
				A&\mathbf{0}_{r,r+1}\\
				\mathbf{0}_{r,r+1}&B\\
				I_{r+1}&I_{r+1}\\
			\end{pmatrix}\label{eq.7},
		\end{eqnarray}where $I_{r+1}$ is the $(r+1)\times(r+1)$ identity matrix, $B$ is the $r\times(r+1)$ matrix
		\begin{eqnarray}
B=			\begin{pmatrix}
				1 & 1 &\cdots&1\\
				\alpha^{<l_1+\hat{t}_1-\hat{t}_2>}&\alpha^{<l_1+\hat{t}_2-\hat{t}_2>}&\cdots&\alpha^{<l_{r+1}+\hat{t}_1-\hat{t}_2>} \\
				\vdots & \vdots& \ddots&\vdots\\
				\alpha^{(r-1)<l_1+\hat{t}_1-\hat{t}_2>}&\alpha^{(r-1)<l_2+\hat{t}_1-\hat{t}_2>}&\cdots&\alpha^{(r-1)<l_{r+1}+\hat{t}_1-\hat{t}_2>} \\
			\end{pmatrix}\nonumber.
		\end{eqnarray} 
In the first step of Eq. \eqref{eq.7}, we swap some of  the last $r+1$ columns, the matrix $B'$ is transformed into $B$ and $P_{\hat{t}_2}'$ is transformed into $P_{\hat{t}_1}'$. In the second step 2 of Eq. \eqref{eq.7}, we swap some of the last $r+1$ rows such that the permutation matrix is transformed into the identity matrix.
Since 
\begin{align*}
\text{rank}\Big(\begin{pmatrix}
			A&\mathbf{0}_{r,r+1}\\
			I_{r+1}&I_{r+1}\\
		\end{pmatrix}\Big)=\text{rank}(A)+\text{rank}(I_{r+1})=2r+1,
\end{align*}
it is sufficient to show that there exists an integer $j\in\{1,2,\cdots,r\}$ such that the matrix 
\[Q_j:=\begin{pmatrix}
			A&\mathbf{0}_{r,r+1}\\
			\mathbf{0}_{1,r+1}&B_j\\
			I_{r+1}&I_{r+1}\\
		\end{pmatrix}
\]
is invertible, where $B_j$ represents the $j$-th row of matrix $B$, i.e., $$B_j=(\alpha^{(j-1)<l_1+\hat{t}_1-\hat{t}_2>},\alpha^{(j-1)<l_2+\hat{t}_1-\hat{t}_2>},\cdots,\alpha^{(j-1)<l_{r+1}+\hat{t}_1-\hat{t}_2>}).$$ 
When $j=1$, we can show that $\text{rank}(Q_1)=2r+1$. In the following, we consider that $j\in\{2,\cdots,r\}$. 
By adding $\alpha^{(j-1)<l_u+\hat{t}_1-\hat{t}_2>}$ times of the $r+1+u$-th row of $Q_j$ to the $r+1$-th row of $Q_j$ for all $u\in\{1,2,\cdots,r+1\}$,
the matrix $Q_j$ is transformed into $$\begin{pmatrix}
			A&\mathbf{0}_{r,r+1}\\
			B_j&\mathbf{0}_{1,r+1}\\
			I_{r+1}&I_{r+1}\\
		\end{pmatrix}.$$ 
Therefore, we have \begin{align*}
\det(Q_j)=\det(\begin{pmatrix}
			A\\
			B_j\\
		\end{pmatrix})=&\det\Big( \begin{pmatrix}
				1 & \cdots&1\\
				\alpha^{l_1}&\cdots&\alpha^{l_{r+1}} \\
				\vdots &\ddots&\vdots\\
				\alpha^{(r-1)l_1}&\cdots&\alpha^{(r-1)l_{r+1}} \\
				\alpha^{(j-1)<l_1+\hat{t}_1-\hat{t}_2>}&\cdots&\alpha^{(j-1)<l_{r+1}+\hat{t}_1-\hat{t}_2>} \\
			\end{pmatrix}\Big)\\
   =&\sum_{u=1}^{r+1}(\alpha^{(j-1)<l_u+\hat{t}_1-\hat{t}_2>}\prod_{\forall1\leq s<v\leq r+1,s\neq t,v\neq t}(\alpha^{l_s}-\alpha^{l_v})).
\end{align*}
If the above determinant is non-zero, then we can repair the erased $2r+2$ symbols; otherwise, we can't repair the erased $2r+2$ symbols.
		
We consider that $\hat{t}_2=m$. With similar proof of the case of $\hat{t}_2\leq m-1$, we can show that we can't repair the erased $2r+2$ symbols.

        \section{Remaining Proof of Theorem \ref{th11}}
        \label{app.2}
	
	\begin{proof}
Suppose that the first $a+1$ nodes $0,1,\ldots,a$ and the last $r$ nodes $k,k+1,\ldots,k+r-1$ are erased. It is sufficient to repair $(a+1)m$ erased data symbols in the first $a+1$ nodes.
According to the construction of our GSRCs, there are total $\sum_{t=0}^{a-1}(m+t)=am+\frac{a(a-1)}{2}$ coded symbols in the surviving nodes which are linear combinations of some erased data symbols. In the following, we show that we can't repair the erased nodes by considering two cases: $m\leq a$ and $m>a$.

When $m\leq a$, according to the repair method of ``\textbf{Repair the first $m$ symbols of node $f_{a+1}$}" in the proof of Theorem \ref{th11}, we can repair the symbols $\{x_{a+\delta-u,u}\}_{u=\delta}^{m-1}$ for all $\delta\in\{0,1,\cdots,m-1\}$. There are $\sum_{\delta=0}^{m-1}(m-\delta)=\frac{m(m+1)}{2}$ symbols.
Therefore, we need to repair the other $(a+1)m-\frac{m(m+1)}{2}$ erased data symbols. However, the number of surviving coded symbols which are linear combinations of some of the erased $(a+1)m-\frac{m(m+1)}{2}$ data symbols is $\sum_{i=a+1}^{a+m}(2a-i)=am-\frac{m(m+1)}{2}$, which is strictly less than the number of the remaining erased data symbols. Therefore, it is impossible to repair all the erased data symbols.

When $m>a$, according to the repair method of ``\textbf{Repair the first $m$ symbols of node $f_{a+1}$}" in the proof of Theorem \ref{th11}, we can repair $\sum_{\delta=m-a}^{m-1}(m-\delta)=\frac{a(a+1)}{2}$ data symbols $\{x_{a+\delta-u,u}\}_{u=\delta}^{m-1}$ for all $\delta\in\{m-a,\cdots,m-1\}$. We still need to repair the other $(a+1)m-\frac{a(a+1)}{2}$ erased data symbols. However, the number of surviving coded symbols which are linear combinations of some of the remaining erased $(a+1)m-\frac{a(a+1)}{2}$ data symbols is $(m-a)a+\sum_{i=m+1}^{a+m-1}(a+m-i)=am-\frac{a(a+1)}{2}$, which is strictly less than the number of the remaining erased data symbols. We can't repair all the erased data symbols.

Therefore, it is impossible to repair the erased $r+a+1$ nodes.
	\end{proof}

        \end{appendices}

	\ifCLASSOPTIONcaptionsoff
	\newpage
	\fi
	
	\bibliographystyle{IEEEtran}
	\bibliography{CNC-v1}
\end{document}